\documentclass[a4paper]{article}

\usepackage[utf8]{inputenc}
\usepackage[english]{babel}
\usepackage{cmap}
\usepackage{amsmath}
\usepackage{amsthm}

\usepackage[left=20mm,right=20mm,top=20mm,bottom=20mm]{geometry}

\theoremstyle{plain}
\newtheorem{lemma}{Lemma}
\newtheorem*{corollary*}{Corollary}
\newtheorem{theorem}{Theorem}
\theoremstyle{definition}
\newtheorem{definition}{Definition}

\author{Vladimir Lysikov\footnote{E-mail: \texttt{lysikov-vv@yandex.ru}}\\{\small Moscow State University}}
\title{On bilinear algorithms for multiplication in quaternion algebras}
\date{}

\begin{document}

\maketitle


\begin{abstract}
We show that the bilinear complexity of multiplication in a non-split quaternion algebra over a field of characteristic distinct from $2$ is $8$. This question is motivated by the problem of characterising algebras of almost minimal rank studied in [1]

This paper is a translation of a report submitted by the author to the XI international seminar "Discrete mathematics and its applications".
\end{abstract}



\begin{definition}
	A sequence $(f_1,g_1,z_1;\dots;f_r,g_r,z_r)$ with $f_k\in U^{*}, g_k\in V^{*}, z_k\in W$ is called a \emph{bilinear algorithm} of length $r$ for a bilinear mapping $\varphi\colon U\times V\to W$ if
	\begin{equation}
		\varphi(u, v) = \sum_{k = 1}^r f_k(u)g_k(v)z_k\quad \forall u\in U, v\in V.
	\end{equation}

	The minimal length of a bilinear algorithm for $\varphi$ is called the \emph{bilinear complexity}, or the \emph{rank}, of $\varphi$ and is denoted by $R(\varphi)$. Algorithms of minimal length are called optimal. 
\end{definition}


Rank of bilinear map has a useful algebraic interpretation. A bilinear map can be thought of as a tensor in $U^{*}\otimes V^{*}\otimes W$. Bilinear algorithms are decompositions of this tensor and the bilinear complexity coincides with the tensor rank as defined in multilinear algebra.


\begin{definition}
	A quaternion algebra over field $F$ of characteristic $\mathop{\mathrm{char}} F\neq 2$ is a $4$-dimensional algebra generated by two elements $i, j$ with identities
	\begin{equation}
		i^2 = p,\ j^2 = q,\ ij = -ji
	\end{equation}
	for some $p,q\in F^{\times}$.
\end{definition}


Any quaternion algebra over $F$ is either isomorphic to the matrix algebra $F^{2\times 2}$ (in this case it is called split) or a noncommutative division algebra. It is known that $R(F^{2\times 2}) = 7$ and $R(H)\geq 8$ for non-split quaternion algebra $H$ [3].


Let $m = \dim U$, $n = \dim V$. We will prove some general results on algorithms of rank $m + n$ for some class of bilinear maps.


\begin{definition}
	We say that an element $u_0\in U$ is (left) $\varphi$-\emph{regular} if the linear map $\varphi(u_0,\cdot)$ is an injection, i. e., if
	\begin{equation}
		\varphi(u_0, v) = 0 \Leftrightarrow v = 0.
	\end{equation}
\end{definition}


\begin{definition}
	A bilinear algorithm $\varphi = \sum_{k = 1}^r f_k\otimes g_k\otimes z_k$ is called \emph{two-component} if the set $\{1,\dots,r\}$ of indices can be partitioned into two sets $I$ and $J$ such that $\{f_i | i\in I\}$ and $\{g_j | j\in J\}$ are bases of $U^{*}$ and $V^{*}$ respectively.
\end{definition}


\begin{lemma}
	If $R(\varphi) = m + n$, $\mathop{\mathrm{lker}} \varphi = \mathbf{0}$ and every basis of $U$ contains a $\varphi$-regular element then all optimal bilinear algorithms for $\varphi$ are two-component.
\end{lemma}


\begin{proof}
Let $\sum\limits_{k = 1}^{m + n} f_k\otimes g_k\otimes z_k$ be an optimal bilinear algorithm for $\varphi$. Since $\mathop{\mathrm{lker}} \varphi = \mathbf{0}$, the functionals $f_1,\dots,f_r$ span $U^{*}$. W.l.o.g. let $f_1,\dots,f_m$ be a basis and $u_1$ be a $\varphi$-regular element of the dual basis $u_1,\dots,u_n$. Since $u_1$ is regular, functionals $g_1,g_{m + 1},\dots,g_{m + n}$ span $V^{*}$.

Case 1. $g_{m + 1},\dots, g_{m + n}$ are linearly independent. Then $I = \{1,\dots, m\}$ and $J = \{m + 1,\dots,m+n\}$ form a partition of $\{1,\dots,m+n\}$ required by definition of two-component algorithm.

Case 2. $g_{m + 1},\dots, g_{m + n}$ are linearly dependent. In this case $\dim  \mathrm{lin}(\{g_{m + 1},\dots, g_{m + n}\}) = n - 1$ and there is a unique (up to a constant) linear dependence of $g_{m + 1},\dots, g_{m + n}$. W.l.o.g. assume that $\sum_{k = 1}^s c_k g_{n + k} = 0$ where all the coefficients $c_k$ are nonzero.

If there is an index $p$ such that $m + 1\leq p\leq m + s$ and $f_p$ has a nonzero first coordinate in basis $f_1,\dots,f_m$ then $I = \{2,\dots,m,p\}$ and $J = \{1,m + 1,\dots, m + n\}\setminus \{p\}$ form a required partition. If all $f_{m + 1},\dots,f_{m + s}$ have zero first coordinates, then we have a contradiction with $\varphi$-regularity of $u_1$, since $\varphi(u_1,v) = 0$ for a nonzero $v\in \ker g_1\cap \bigcap_{k = n + s + 1}^{m + n} \ker g_k$.
\end{proof}


\begin{lemma}
	A two-component bilinear algorithm for $\varphi$ exists if and only if there are bases $(u_1,\dots,u_m)$ and $(v_1,\dots,v_m)$ of spaces $U$ and $V$ resp. and collections $(z'_1,\dots,z'_m)$, $(z''_1,\dots,z''_n)$ with $z'_i, z''_j\in W$ such that the following condition holds:
	\begin{equation}
	\varphi(u_i, v_j)\in \mathrm{lin}(\{z'_i, z''_j\}).
	\end{equation}
\end{lemma}



\begin{proof}
Any two-component bilinear algorithm for $\varphi$ can be written as
\begin{equation}
\label{eqn5}
\varphi = \sum_{i = 1}^m f_i\otimes (\sum_{j = 1}^n \lambda_{ij} g_j)\otimes z'_i + \sum_{j = 1}^n (\sum_{i = 1}^m \mu_{ij} f_i)\otimes g_j\otimes z''_j,
\end{equation}
where $(f_i)$ and $(g_j)$ are bases of $U^{*}$ and $V^{*}$ respectively. Let $(u_i)$ and $(v_j)$ be the bases dual to $(f_i)$ and $(g_j)$. Then it follows that
\begin{equation}
\label{eqn6}
\varphi(u_i, v_j) = \lambda_{ij} z'_i + \mu_{ij} z''_j.
\end{equation}

Conversely, if \eqref{eqn6} holds then \eqref{eqn5} is a two-component algorithm for $\varphi$.
\end{proof}


\begin{corollary*}
Let $A$ be a local algebra, $\dim A = n$, and $R(A) \geq 2n$. Equality $R(A) = 2n$ holds iff there are bases $(u_1 = 1, u_2,\dots u_n)$, $(v_1 = 1, v_2,\dots v_n)$ and collections $(z'_i, \dots, z'_n)$, $(z''_1,\dots,z''_n)$ such that
\begin{equation}
u_i v_j \in \mathrm{lin}(\{z'_i, z''_j\}).
\end{equation}
\end{corollary*}


\begin{theorem}
	Let $F$ be a field of characteristic $\mathop{\mathrm{char}} F\neq 2$ and $H$ be a non-split quaternion algebra over $F$. Then $R(H) = 8$.
\end{theorem}


\begin{proof}
Apply the preceding corollary with the following bases and collections:
\addtocounter{equation}{1}
\begin{equation}
(u_1, u_2, u_3, u_4) = (v_1,v_2,v_3,v_4) = (1, i, j, k)\tag{\theequation a}
\end{equation}
\begin{equation}
(z'_1, z'_2, z'_3, z'_4){=}(1 + \alpha i + \beta j + \gamma k, 1 + \alpha i - \beta j - \gamma k, 1 - \alpha i + \beta j - \gamma k, 1 - \alpha i - \beta j + \gamma k)\tag{\theequation b}
\end{equation}
\begin{equation}
(z''_1, z''_2, z''_3, z''_4){=}(1 - \alpha i - \beta j - \gamma k, 1 - \alpha i + \beta j + \gamma k, 1 + \alpha i - \beta j + \gamma k, 1 + \alpha i + \beta j - \gamma k)\tag{\theequation c}
\end{equation}
where $\alpha,\beta,\gamma\in F$, $\alpha\neq 0$, $\beta \neq 0$, $\gamma \neq 0$.
\end{proof}



By considering bilinear algorithms arising from different constants $\alpha,\beta,\gamma$ we also managed to prove the following fact about equivalence of optimal algorithms in the sense of de Groote [2].

\begin{theorem}
	Let $F$ be a field of characteristic $\mathop{\mathrm{char}} F\neq 2$ and $H$ be a non-split quaternion algebra over $F$. There are infinitely many de Groote equvalence classes of optimal bilinear algorithms for the multiplication in $H$.
\end{theorem}



The author is grateful to Prof.~V.~B.~Alekseyev for his attention to this research.

This research was supported by RFBR grant 12-01-91331-DFG-a.

\smallskip
{\bf Bibliography}

\begin{enumerate}
\item M.~Bl\"{a}ser, A.~M.~de~Voltaire. Semisimple algebras of almost minimal rank over the reals // Theoretical Computer Science.~--- 2009.~--- vol.~410, no.~50.~--- pp.~5202--5214.
\item H.~F.~de Groote. On varieties of optimal algorithms for the computation of bilinear mappings I // Theoretical Computer Science.~--- vol.~7, no.~1.~--- pp.~1--24.
\item P.~B\"{u}rgisser,M.~Clausen,M.~A.~Shokrollahi. Algebraic complexity theory.~---Springer~Verlag, 1997.
\end{enumerate}
\end{document}